\documentclass[11pt,draftclsnofoot,onecolumn,nofoot]{IEEEtran}
\usepackage{amsmath, amsthm, amssymb}
\usepackage{graphics}
\usepackage{graphicx}
\usepackage{epsfig}

\newtheorem{theorem}{Theorem}
\newcounter{Lemma}
\newtheorem{lemma}[Lemma]{Lemma}

\newenvironment{definition}[1][Definition]{\begin{trivlist}
\item[\hskip \labelsep {\bfseries #1}]}{\end{trivlist}}

\begin{document}

\title{A Lower Bound on the Capacity of Wireless Erasure Networks with Random Node Locations}
\author{Rayyan G. Jaber and Jeffrey G. Andrews\footnote{R. Jaber and J. Andrews are with the Wireless Networking and Communications Group (WNCG) of the Electrical and Computer Engineering
Department, The University of Texas at Austin, Austin, TX, 78712-0240 USA (email:
\{jaber, jandrews\}@ece.utexas.edu). The contact author is J. Andrews.  This research was
supported by NSF grant nos. 0634979 and 0643508 and the DARPA IT-MANET program,
grant no. W911NF-07-1-0028.  Manuscript date: \today.}} \maketitle

\begin{abstract}
In this paper, a lower bound on the capacity of wireless ad hoc erasure networks is derived in closed form in the canonical case where $n$ nodes are uniformly and independently distributed in the unit area square. The bound holds almost surely and is asymptotically tight. We assume all nodes have fixed transmit power and hence two nodes should be within a specified distance $r_n$ of each other
to overcome noise. In this context, interference determines outages, so we model each transmitter-receiver pair as an erasure channel with a broadcast constraint, i.e. each node can transmit only one signal across all its outgoing links.
A lower bound of $\Theta(n r_n)$ for the capacity of this class of networks is derived. If the broadcast constraint is relaxed and each node can send distinct signals on distinct outgoing links, we show that the gain is a function of $r_n$ and the link erasure probabilities, and is at most a constant if the link erasure probabilities grow sufficiently large with $n$.  Finally, the case where the erasure probabilities are themselves random variables, for example due to randomness in geometry or channels, is analyzed. We prove somewhat surprisingly that in this setting, variability in erasure probabilities increases network capacity.
\end{abstract}

\section{Introduction}

Determining the capacity regions of multiuser wireless networks is an open problem in
general~\cite{And07}. Previous work develops approximations to and descriptions of the
network capacity in different settings, with several different approaches \cite{GupKum00,
TouGol03, OzgLevTse07, ElGMam04, DanGow06}.  Gupta and Kumar~\cite{GupKum00} began the popular trend of
characterizing the fundamental limits on the throughput of such networks with scaling
laws. In particular, they prove that a sublinear sum rate throughput of
$\Theta(~\sqrt{\frac{n}{\log{n}}})$ is achievable, where $n$ is the number of
nodes in the network, if the nodes are uniformly and independently distributed in a unit
area, each transmitting to a randomly chosen destination~\cite{GupKum00}.  Scaling laws
have been further developed in a number of scenarios
\cite{XieKum04,XueXie05,LevTel05,JovVis04} and the $\log n$ factor for the lower bound
was proven to be superfluous~\cite{FraDou07}.   Recently, Ozgur~\emph{et~al.} argued
that a linear scaling $O(n^{1-\epsilon})$ may be approached in the case of hierarchal
cooperation ~\cite{OzgLevTse07}, but this increases delay and in any case may not
change the underlying scaling in real channels \cite{Fra07}.  Similarly, with randomized
mobility and unbounded delay, a ``postman" model of packet delivery can be employed to
get linear, i.e. $O(n)$, scaling \cite{GroTse02}, which has led to studies on
throughput-delay trade-offs ~\cite{ElGMam04}.  A common feature of all this work is that
the preconstants to the scaling laws are not computable, which has rendered the
quantitative results generated from these approaches to be coarse.  In many cases, this
has impaired qualitative improvement in the design of distributed wireless networking
protocols.  Some new approaches seem necessary to quantifying the network capacity.  The
goal of this paper is to advance such an approach, showing how straightforward tools from
random geometric graph theory can be used to replicate the aforementioned scaling laws,
while providing further precision on the preconstants.

\subsection{Erasure Networks}
Erasure networks characterize transmission links in a wireless ad hoc network by
assigning an erasure probability to each potential connection between nodes in the
network~\cite{Thom}. From a practical perspective, erasure events correspond to packet drops or
temporary outages and are a reasonable metric for characterizing a channel with a certain
bit rate.  Dana~\emph{et~al.} recently derived elegant cut-set bounds to characterize the
capacity of wireless ad hoc erasure networks under a set of reasonable
assumptions~\cite{DanGow06}. Their result, however, is independent of the network
topology and geometry of the node locations, which are the most important effect in
determining the erasure probabilities and traffic flows in the network. Instead, the
capacity was cast as an optimization program that involves minimizing a (nonlinear) cut-set expression over a set whose size is exponential in the number of nodes $n$. Beside
the inherent difficulties in evaluating an exponentially large number of cut-sets even in
moderate sized networks ($n=50$ is computationally very intensive), this result does not
reveal how the network capacity depends on parameters such as number of nodes, the
erasure probabilities and transmission range.

The present paper thus aims to establish a model for wireless erasure ad hoc networks
that captures node topology, physical layer parameters, and develops tight bounds in
closed form for the end-to-end throughput.  We place $n$ nodes uniformly and
independently in the unit square $[0,1]^2$, each of which can communicate with other
nodes within distance $r_n$ through wireless broadcast erasure links with constant
erasure probability $\gamma_n$. We first consider the single source single destination and then generalize it to multiple sources multiple destinations as follows. We assume the set of intended transmitters form a linear fraction of the nodes and so do the intended receivers. We also assume that each of these sets can cooperate among each other. The remaining nodes can relay messages within their transmission range $r_n$. Finally, the failure events of transmissions across distinct links are assumed independent and happen with probability $\gamma_n$.

\subsection{Main Results}
The main result of the paper is given in Theorem~\ref{thm:mainthm}, which provides
a closed-form lower bound on the capacity of a wireless erasure network.  This provides a
scaling law on the network sum capacity of at least $\Theta(n r_n)$ for an
\emph{arbitrary} set of transmitters to an \emph{arbitrary} set of receivers, each of which contains a linear fraction of the number of nodes and is
allowed to cooperate their transmissions and receptions.  We further show in Section~\ref{sec:tight} that the bound
is tight, in the sense that there exists a particular choice of source nodes and
destination nodes for which the sum rate capacity is within a (small) constant factor
from the proven lower bound.  Thus, for the critical connectivity radius of $r_{\rm c} =
\Theta\left(\sqrt{\frac{\log n}{n}}\right)$, our lower bound scales as
$\Theta(\sqrt{n\log{n}})$, consistent with \cite{GupKum00} up to a $\log{n}$ factor although with a quite
different network model.

If the broadcast constraint is relaxed to allow transmitting nodes to send distinct
messages across their outgoing links (multicast), we prove in Theorem~\ref{thm:relaxingbroadcast} that the network capacity
increases to $\Theta(n^2 r_n^3)$ if the erasure probabilities are constant with $n$. This
gain evaluates to a factor of $\log{n}$ for $r=r_{\rm c}$. However, as the erasure
probabilities increase with the number of nodes $n$, say due to increased interference, then the gain
due to multicast starts to decrease. At the critical case, where the erasure probabilities scale as of $1-1/\Theta(\log{n})$ the gain due to multicast is at most a constant.
  Finally, we prove that if the erasure
probabilities are not constant even for fixed $n$, but are random variables instead as
would be the case in a network with fading and random node locations, then this
variability in erasure probabilities actually increases the network capacity as proven in Lemma~\ref{lemma:randomearsures}. The
intuition behind this initially surprising result is that only one successful
(non-erased) transmission is needed to traverse a cut, so variability provides
statistical diversity that improves the chances of at least one successful transmission.

\subsection{Organization}
Preliminaries on notations, especially with respect to random geometric graphs are given in Section~\ref{Notation}. The modeling assumptions are stated in Section~\ref{assumptions} and the capacity cut set bound identified in Section~\ref{capCutSet} is used as a suitable metric in the setting of ad hoc wireless erasure networks. Section~\ref{sec:randomGeometricGraphsAndGrids} draws the analogy between random geometric graphs and the deterministic grid. Section~\ref{relatingtoadhoc} then establishes the desired lower bound of $\Theta(nr_n)$ and proves it tight. Section~\ref{sec:multicast} proves that a gain of $nr_n^2$ is achieved by relaxing the broadcast constraint. It also proves variability in erasure probabilities increases network capacity. Section~\ref{sec:conclusion} concludes the paper.

\section{Preliminaries}
\label{Notation}
\subsection{Notation}
Throughput this paper, sets are denoted by calligraphic alphabet (e.g. $\mathcal{A}, \mathcal{B}, \mathcal{C}$), $|\mathcal{X}|$ and $\mathcal{X}^c$ denote the cardinality and the complement of set $\mathcal{X}$ respectively. The logarithm $\log{x}$ denotes the natural logarithm of a positive real number $x$. Some parameters of the network model will depend on the number of nodes $n$ in the network. These parameters are subscripted $n$. For example, $r_n$ denotes the transmission radius of a node when the network has $n$ nodes. The subscript $n$ might be dropped when it is implied in context. When the number of nodes in the network is implied, a subscript might be used to denote a sequence of nodes or links. For example, in a network of $n$ nodes, $v_1, v_2,\ldots\, v_k$ denotes a sequence of $k$ nodes.

\subsection{Definitions for Random Geometric Graphs}
For two points $x,y \in \mathbb{R}^2$, the distance between $x$ and $y$ is $||x-y||_{\infty} = \max\left\{x_1-y_1, x_2-y_2\right\}$ where $(x_1, x_2)$ and $(y_1,y_2)$ are the coordinates of $x$ and $y$ respectively. This measure of distance simplifies the analysis of the lower bound presented in section~\ref{relatingtoadhoc}. Similar results to those proven in this paper follow if the Euclidian distance $||x-y||_2 = \sqrt{(x_1-y_1)^2+(x_2-y_2)^2}$ is used. Given two real valued functions $f(n)$ and $g(n)$, $f(n)=\Theta(g(n))$ if and only if there exists positive constants $c_1, c_2$ and $n_0$ such that $c_1 g(n) \leq f(n) \leq c_2 g(n)$ for all $n\geq n_0$.\\
Let $\mathcal{V}_n$ be a Bernoulli point process consisting of $n$ points (or nodes) independently and identically distributed in the unit square $[0,1]^2$ and let $r_n$ be a positive real number. For every integer $n \geq 1$, we can construct the graph $\mathcal{G}_n=\mathcal{G}(\mathcal{V}_n,r_n)$ as the graph on $n$ vertices, associated with $\mathcal{V}_n$ and the set of directed edges $\mathcal{E}_n\subset \mathcal{V}_n \times \mathcal{V}_n$ is characterized as:
 $$\mathcal{E}_n= \left\{(u,v) | u,v \in \mathcal{V}_n \textrm{ and } ||u-v||_\infty \leq r_n\right\}$$
The graph $\mathcal{G}_n$ is said to be a \emph{random geometric graph} and it can be completely parameterized by $n$ and $r_n$, where $r_n$ is called the transmission radius of the nodes~\cite{Pen03}. Fig.~\ref{fig:network} shows an example of a random geometric graph with 50 nodes and transmission radius of 0.2.
\begin{figure}
\centering
\centerline{\epsfig{figure=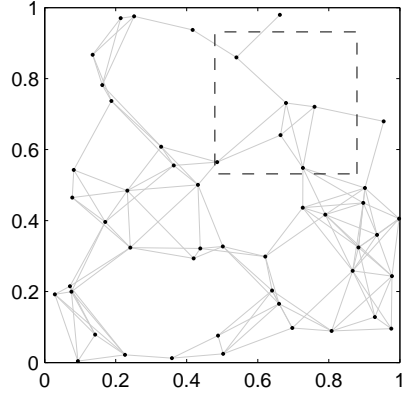,scale=.75}}
\caption{A random geometric graph with $n=50$ nodes and transmission radius $r=0.2$. The dotted square represents the transmission range of the node at its center using the $L\infty$ norm for distances.}
\label{fig:network}
\end{figure}
The following graph theoretic definitions are defined for every $n$, and hence the $n$ subscript will be dropped for convenience~\cite{Wes96}. For nodes $u,v \in \mathcal{V}$, $u$ is said to be \emph{connected} to $v$ if and only if $(u,v)\in\mathcal{E}$. For each node $v \in \mathcal{V}$, $\mathcal{N}_O(v)$ is the set of edges leaving from $v$. Formally
\begin{eqnarray*}
\mathcal{N}_O(v) &=& \left\{ (v,u)|(v,u)\in\mathcal{E} \right\}.
\end{eqnarray*}
Given two disjoint subsets $\mathcal{S}, \mathcal{D} \subset \mathcal{V}$, an $\mathcal{S}-\mathcal{D}$ \emph{cut} is a partition of $\mathcal{V}$ into subsets $\mathcal{V}_\mathcal{S}$ and $\mathcal{V}_\mathcal{D}= \mathcal{V}_\mathcal{S}^c$ such that $\mathcal{S}\subseteq \mathcal{V}_\mathcal{S}$ and $\mathcal{D}\subseteq\mathcal{V}_\mathcal{S}^c$. The $\mathcal{S}$-set $\mathcal{V}_\mathcal{S}$ (or $\mathcal{D}$-set $\mathcal{V}_\mathcal{D}$) determines the cut uniquely. For the $\mathcal{S}-\mathcal{D}$ cut given by $\mathcal{V}_\mathcal{S}$, the \emph{cut-set} $\left[ \mathcal{V}_\mathcal{S}, \mathcal{V}_\mathcal{D} \right]$ is the set of edges going from the $\mathcal{S}$-set to $\mathcal{D}$-set, i.e.,
$$\left[ \mathcal{V}_\mathcal{S}, \mathcal{V}_\mathcal{D} \right] = \left\{ (u,v) | (u,v) \in \mathcal{E}, u \in \mathcal{V}_\mathcal{S}, v \in \mathcal{V}_\mathcal{D} \right\}$$
We also define $\mathcal{V}_\mathcal{S}^*$ as the set of nodes in the $\mathcal{V}_\mathcal{S}$-set that has at least one of its outgoing edges in the cut-set. That is $$\mathcal{V}_\mathcal{S}^* = \left\{ v | \exists u~\textrm{s.t.}~(v,u)\in \left[\mathcal{V}_\mathcal{S},\mathcal{V}_\mathcal{D}\right]\right\}$$

In a given graph, a \emph{path} from node $u_1$ to node $u_k$ is a sequence of edges in $\mathcal{E}$: $(u_1,u_2), (u_2,u_3),\ldots, (u_{k-1},u_{k})$. We refer to a path by its corresponding sequence of nodes $u_{1}u_{2}\ldots u_{k}$. There might be multiple paths from node $u$ to node $v$. If there exists at least one path from every node $u$ to every other node $v$, the graph $\mathcal{G}$ is said to be \emph{connected}. Otherwise, it is said to be \emph{disconnected}. The graph in Fig.~\ref{fig:network} is connected.

\begin{figure}
\centering
\centerline{\epsfig{figure=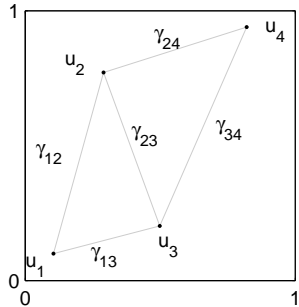,scale=.75}}
\caption{A simple network with 4 nodes. The erasure probabilities are denoted on the edges and are assumed to be symmetric, that is $\gamma_{ij} = \gamma_{ji}$ for simplicity.}
\label{fig:simple_network}
\end{figure}

Fig.~\ref{fig:simple_network} is an instance of a random geometric network with 4 nodes and transmission radius of $0.75$. The edges are labeled with the erasure probabilities of their corresponding links. Erasure probabilities are assumed symmetric for simplicity, that is: $\gamma_{ij}=\gamma_{ji}$. If $u_1$ is the source node and $u_4$ is the destination node, then there are 4 possible cuts depending on which side of the cut the nodes $u_2$ and $u_3$ are placed. Table~\ref{tbl:cutsets} lists these 4 possible cuts. The function $C(\mathcal{V}_\mathcal{S})$ that appears in the last column of the table is defined later, in equation~\ref{eqn:capacity}. As we shall see in section~\ref{capCutSet}, $C(\mathcal{V}_\mathcal{S})$ is an upper bound to how much information can flow across the cut $\mathcal{V}_\mathcal{S}$.

\begin{table}
  \centering
  \caption{Possible Cut Sets}\label{tbl:cutsets}
  \begin{tabular}{|l|l|l|l|l|}
    \hline
    $\mathcal{V}_{\mathcal{S}}$ & $\mathcal{V}_{\mathcal{S}}^*$ & $[\mathcal{V}_{\mathcal{S}},\mathcal{V}_{\mathcal{D}}]$ & $C(\mathcal{V}_{\mathcal{S}})$ \\
    \hline
    $\{u_1\}$           & $\{u_1\}$ & $\{(u_1,u_2),(u_2,u_3)\}$ & $1-\gamma_{12}\gamma_{23}$  \\
    $\{u_1, u_2\}$      & $\{u_1,u_2\}$  & $\{(u_1,u_3),(u_2,u_3),(u_2,u_4)\}$ & $(1-\gamma_{13})+(1-\gamma_{23}\gamma_{24})$ \\
    $\{u_1, u_3\}$      & $\{u_1,u_3\}$  & $\{(u_1,u_2),(u_3,u_2),(u_3,u_4)\}$ & $(1-\gamma_{12})+(1-\gamma_{32}\gamma_{34})$ \\
    $\{u_1, u_2, u_3\}$ & $\{u_2,u_3\}$ & $\{(u_2,u_4),(u_3,u_4)\}$ & $(1-\gamma_{24})+(1-\gamma_{34})$  \\
    \hline
  \end{tabular}
\end{table}

\section{Modeling Assumptions}
\label{assumptions}
This section specifies a reasonable model for a wireless network that is simple and tractable yet resembles the actual physical system. It also scrutinizes the underlying assumptions and questions their validity.

\subsection{Nodes and Links}
We consider the case of $n$ nodes $\mathcal{V}_n$ independently and uniformly distributed in the unit square $[0,1]^2$ forming the binomial point process $\mathcal{V}_n$. This distribution is equivalent to conditioning a stationary Poisson point process on having exactly $n$ points in the unit square~\cite{Kin93}. Previous work has shown that stochastic geometry based on Poisson point processes can capture key features of wireless networks~\cite{BacKle97,Hae06,VenHae06,StoKen96,BacPage,BacZuy97,Bac06,ChaHan01}. Results about random geometric graphs where the nodes are uniformly distributed often yield similar results when the nodes are distributed according to a Poisson point process, by what is referred to as ``Poissonization''~\cite{Pen03}. In this paper, we focus on the case of the binomial point process because this setting has been a canonical example in modeling node locations~\cite{GupKum00, OzgLevTse07}. This will hence allow us to compare the main results of this paper with previous results.\\

Let $\mathcal{S}_n, \mathcal{D}_n \subset \mathcal{V}_n$ be two arbitrary but disjoint subsets of $\mathcal{V}_n$, denoting the sets of intended transmitters and receivers respectively. We also assume that $|\mathcal{S}_n|=\alpha_1 n$ and $|\mathcal{D}_n|=\alpha_2 n$ for some positive real constants $\alpha_1$ and $\alpha_2$. We assume all nodes have fixed transmit power (for a fixed $n$). Since the transmit power is finite, and because of the decay of power with distance ($d^{-\alpha}$ in path loss models), two nodes should be close enough to each other for the signal to noise and interference ratio at the receiver to exceed the minimum threshold needed for successful transmission. The signal to interference and noise ratio is assumed to be negligible at distances farther than $r_n$ from the transmitter. It is hence natural to consider the random geometric graph $\mathcal{G}_n = \mathcal{G}(\mathcal{V}_n, r_n)$.
In fact, random geometric graphs have been extensively used as a model of large
wireless networks \cite{Jia04, Dia07}.

Of course, we do not expect all transmissions to be successful between connected nodes. Indeed, due to fading, noise and possibly interference, the links are not perfect links and are modeled as erasure channels. For every link $(u,v)\in$ in the set of edges $\mathcal{E}_n$, denote by $\gamma_{uv}$ its erasure probability.
%The above construction of the network and its corresponding links induces a random geometric graph where the vertices of the graph are the nodes and weighted edges are the links connecting nodes that are within transmission distance $r_n$ of each other other. The weights denote the erasure probabilities of the corresponding channels. Random geometric graphs capture the location of the nodes and hence are better models for problems where connectivity is related to the location of the nodes.
%Since the transmitter signal power decays with distance $d$, as $d^{-\alpha}$ in path loss models where $\alpha$ is the path loss exponent, it is much easier for points close to each other to communicate than points that are far away from each other. Moreover, the wireless nature of transmissions induces an ``interference field'' that is location dependent. These events can be nicely captured in geometric graphs.

\subsection{Network}
We adopt a similar network model as that described in~\cite{DanGow06}. We assume that the nodes form a wireless erasure multi-hop ad hoc network, so that the network comprises the following salient features:\\
{\bf Wireless:}
Each node can only broadcast its message to all its neighboring nodes whenever it chooses to transmit. Section~\ref{relatingtoadhoc} investigates relaxing this constraint and analyzes the associated gain in throughput.\\
{\bf Erasure:}
A transmission on link $(u,v) \in \mathcal{E}_n$ can fail with probability $\gamma_{uv}$ for some $0\leq \gamma_{uv}\leq 1$. Currently, we assume that erasures across distinct links are statistically independent for tractable analysis. This is an idealized assumption due to interference: if a transmission for a certain receiver failed, then it is more likely that other transmissions to neighboring receivers have failed too. Moreover, it is also assumed, as in~\cite{DanGow06}, that messages received at a node from different incoming links do not interfere. This slightly contradicts the wireless assumption unless the network utilizes an appropriate interference avoidance mechanism.\\
{\bf Multi-hop:}
Transmissions are multi-hop so any node can relay packets from one node to another neighboring node.\\
{\bf Ad hoc:}
The network is fully distributed and does not utilize a preexisting infrastructure or central base stations. The set of source nodes and destination nodes however are assumed to be capable of cooperating in their transmissions and receptions respectively.\\
{\bf Cooperative Network:}
Since we are after the \emph{capacity} of such a network, the \emph{maximum} achievable rate from a transmitter to a receiver, we are inherently assuming that the nodes may cooperate to ensure this high rate~\cite{Thom}. As assumed in~\cite{DanGow06}, error locations on each link are available to the destination as side information. This slightly contradicts the ad hoc assumption since the overhead to achieve this cooperation is likely non-negligible, but accounting for it is postponed to future work.\\

This network topology is analytically tractable. Our work exploits many similarities of this topology with the simple deterministic grid topology to derive bounds on end-to-end throughput.
 %The graph below shows one such bound: an approximation of the capacity as a function of the number of nodes and the %nodes transmission radius.

\section{Capacity of Wireless Erasure Channels}
\label{capCutSet}
Under the assumptions stated above, the capacity of single source, single destination wireless erasure networks is elegantly characterized in~\cite{DanGow06}. It is stated as a cut set bound which has a max-flow/min-cut interpretation which practically identifies the ``bottleneck'' in the network.
In particular, for any source node $s$ and destination node $d$, let $\mathcal{S}=\{s\}$, $\mathcal{D}=\{d\}$, and any $\mathcal{V}_\mathcal{S}$-cut of the nodes, the capacity of the network is upper bounded by~\cite{DanGow06}:
\begin{equation}
\label{eqn:capacity}
C(\mathcal{V}_\mathcal{S})=\sum_{u \in \mathcal{V}_\mathcal{S}^*}\left(1-\prod_{v:(u,v)\in[\mathcal{V}_\mathcal{S},\mathcal{V}_\mathcal{D}]}\gamma_{ij}\right)
\end{equation}

And the capacity $C$ of the network is exactly the minimum of the above expression over all possible cut sets $\mathcal{V}_\mathcal{S}$~\cite{DanGow06}:
\begin{equation}
\label{eqn:mincut}
C = \min_{\mathcal{V}_\mathcal{S}: \mathcal{V}_\mathcal{S} \textrm{ an } \mathcal{S}-\mathcal{D} \textrm{ cut}} C(\mathcal{V}_\mathcal{S})
\end{equation}
The expression in~(\ref{eqn:mincut}) is proved to be an achievable upper limit.

The capacity cut set bounds for the network in Fig.~\ref{fig:simple_network} are calculated in Table~\ref{tbl:cutsets}. The capacity of that network is hence given by the expression:
$$C = \min\left\{1-\gamma_{12}\gamma_{23}, 1-\gamma_{13}+1-\gamma_{23}\gamma_{24}, 1-\gamma_{12}+1-\gamma_{32}\gamma_{34}, 1-\gamma_{24}+1-\gamma_{34}\right\}$$

%\begin{figure}
%\centering
%\includegraphics[scale=0.8]{simple_network_details.png}
%\caption{A simple network with 4 nodes and a cut set}
%\label{fig:network}
%\end{figure}

%Observe the multiplicative term $\prod_{j\in N(i)\cap T}\epsilon_{ij}$. Intuitively, this represents the probability that all transmissions from set $S$ to $T$ via node $i$ at a particular time slot fail. The assumption that the transmissions are independent is obviously required here. In fact, if the cut cost was linear and of the form $\sum_{i \in S}\sum_{j \in N(i)\cap T} \gamma_{ij}$, then it can be evaluated easily since it will be an instance of the classic min-cut/max-flow problem; however, it is not obvious how to evaluate the original (nonlinear) cut set bound efficiently. An exhaustive search over all possible cuts will of course return the answer but the runtime is exponential in the number of nodes. Even for relatively small networks of 40 nodes, such naive search is not practical. This motivated us to explore alternative approaches.
Although~(\ref{eqn:mincut}) characterizes the capacity of general networks exactly, it is not obvious to what it evaluates to in practical scenarios, such as the one we consider in this paper, where the nodes are independently and uniformly distributed in space (e.g. in the unit square), each having a fixed transmission radius, with multiple sources and multiple destinations\footnote{More details about the multiple source multiple destination case in Theorem~\ref{thm:mainthm}.}. For the single source single destination case, there are $2^{n-2}$ possible cut sets and evaluating~(\ref{eqn:capacity}) for every one of them is not practical even for relatively small networks. Moreover, the effect of the number of nodes $n$ and transmission radius $r_n$ on the capacity of the network is not clear from~(\ref{eqn:mincut}). Our goal is to identify a lower bound for the capacity that holds almost surely under the assumptions stated in Section~\ref{assumptions} and that highlights the effect of physical layer parameters such as transmission radius $r_n$ and erasure probabilities.

\section{Random Geometric Graphs and Grids}
\label{sec:randomGeometricGraphsAndGrids}
The core of the subsequent analysis is based on random geometric graph theory. The analysis can be divided into three stages:
\begin{enumerate}
\item
An analogy between the random network topology and the deterministic grid is derived.
\item
Relevant combinatorial properties of the grid topology are explored.
\item
These properties are translated back to the probabilistic setting to conclude a lower bound on the capacity of the ad hoc wireless network.
\end{enumerate}

\subsection{Connectivity of Random Ad Hoc Networks and Their Analogy to Grids}
Since the nodes are uniformly and independently distributed, it is intuitive to assume that if the number of nodes is large, then the nodes will be somewhat evenly distributed across the unit square. Indeed, this turns out to be the case as formalized below with the notion of $\epsilon$-niceness.

\begin{definition}
\label{def:epsnice}\emph{$\epsilon$-niceness of a random geometric graph}~\cite{Dia01}.
Consider a random geometric graph $\mathcal{G}_n$ of $n$ nodes $\mathcal{V}_n$ in the unit area square $[0,1]^2$ and connectivity radius $r_n$. Partition the square into $4 \lceil 1 / r_n \rceil^2$ smaller square cells with a side length of $1/(2\lceil 1/r_n \rceil)$. Let $\epsilon \in (0,1)$. The random geometric graph $\mathcal{G}_n$ is said to be $\epsilon$-nice if and only if the number of nodes in each cell is at least $\left(1-\epsilon \right)\frac{1}{4}n r_n^2$ and is at most $\left(1+\epsilon\right)\frac{1}{4}n r_n^2$.
\end{definition}

\begin{theorem}
\label{thm:epsilonnice}
A random geometric graph is $\epsilon$-nice almost surely if $r_n > \frac{6}{\epsilon} \sqrt{\frac{\log{n}}{n}}$ for sufficiently large $n$.
\end{theorem}
\begin{proof}
A simple proof based on the Chernof bound of binomial random variables is presented in Lemma 5.1 in~\cite{Dia01} when $\lim_{n\rightarrow\infty}\frac{r_n^2}{\log{n}/{n}}=\infty$ but it actually suffices for that proof that $r_n > \frac{6}{\epsilon} \sqrt{\frac{\log{n}}{n}}$ for sufficiently large $n$.
\end{proof}

\begin{lemma}
\label{lemma:connected}
A random geometric graph $\mathcal{G}_n$ is almost surely connected if $r_n > 7 \sqrt{\frac{\log{n}}{n}}$ for sufficiently large $n$.
\end{lemma}
\begin{proof}
Assume $r_n > 7 \sqrt{\frac{\log{n}}{n}}$ for sufficiently large $n$. Then by theorem~\ref{thm:epsilonnice}, $\mathcal{G}_n$ is $\epsilon$-nice almost surely for $\epsilon=\frac{6}{7}$. Therefore, if we dissect the unit square as described in the definition of $\epsilon$-niceness, every square cell contains at least one node for sufficiently large $n$. Since the side length of each square cell is $\frac{r_n}{2}$, each point is connected to all the points contained in neighboring cells, including diagonals. This is sufficient to establish a path from any node to any other node.
\end{proof}

In fact, the threshold $\sqrt{\frac{\log{n}}{n}}$ is asymptotically tight as proven in the following theorem.
\begin{theorem}
\label{thm:disconnected}
Given a random geometric graph $\mathcal{G}_n$, there exists a constant $c_1 > 0$, independent of $n$, such that if $r_n < c_1 \sqrt{\frac{\log{n}}{n}}$ for sufficiently large $n$, then $\mathcal{G}_n$ is almost surely disconnected.
\end{theorem}
\begin{proof}
 The results of~\cite{GupKum98} imply the above theorem. The distance metric in~\cite{GupKum98} is the Euclidean norm. But note that if a graph $\mathcal{G}(\mathcal{V}_n,r_n)$ is disconnected under the Euclidean norm metric then the graph $\mathcal{G}'\left(\mathcal{V}_n,r_n/\sqrt{2}\right)$ is also disconnected under the $L_\infty$ norm since the set of edges of $\mathcal{G}'$ is a subset of that of $\mathcal{G}$.
\end{proof}

In ad hoc wireless networks, the transmission radius $r_n$ is often limited by peak power constraints, the rapid power decay with distance and by interference constraints. A node requires less power to broadcast with a smaller transmission radius, since in path-loss models for example, power decays with the distance $d$ as $d^{-\alpha}$ where $\alpha$ is the path-loss exponent. A smaller transmission radius, or equivalently smaller transmission power, will also cause less interference to neighboring nodes.
Lemma~\ref{lemma:connected} and Theorem~\ref{thm:disconnected} together identify the sharp connectivity threshold $ \sqrt{\frac{\log{n}}{n}}$ as the asymptotically smallest transmission radius that ensures connectivity. For this reason, connectivity radii satisfying $r_n = c \sqrt{\frac{\log{n}}{n}}$ for sufficiently large $n$, will be of special interest in subsequent results, whereby the constant $c$ is assumed sufficiently large to ensure connectivity.

%\vspace{-1em}
\subsection{Grid Inequalities}
\label{gridinequalities}
Since we can carefully treat geometric graphs like grids, it makes sense to explore structural properties on grids and apply them to geometric graphs. Here, we present one such property and will demonstrate an application to it later when analyzing the capacity cut on a random geometric graph in Section~\ref{relatingtoadhoc}.

\begin{lemma}
\label{eqn:isoperimetry}
Let $(\mathcal{A},\mathcal{B})$ be a partition of $\left\{1,2,\ldots,m\right\}^2$ for some integer $m$.
Define the boundary length $\partial_{\mathcal{A},\mathcal{B}}$ to be the number of elements of $\mathcal{A}\times \mathcal{B}$ that are neighbors, including diagonals. Then for any $m\geq 3$ and any partition $(\mathcal{A},\mathcal{B})$, $\partial_{\mathcal{A},\mathcal{B}} \geq 3\min\left\{\sqrt{\left|\mathcal{A}\right|}, \sqrt{\left|\mathcal{B}\right|}\right\}$~\cite{Dia01}.
\end{lemma}

A combinatorial proof can be found in Section 4 in~\cite{Dia01}. The lemma is illustrated in Fig.~\ref{fig:isoperimetry} on a $4\times 4$ grid. %This inequality can be tight, for example the degenerate case where $\mathcal{A}$ consists of a single corner point, say $\{(1,1)\}$ and the set $\mathcal{B}=\mathcal{A}^c$. In that case, $\partial_{\mathcal{A},\mathcal{B}} = 3\min\left\{\sqrt{\left|\mathcal{A}\right|}, \sqrt{\left|\mathcal{B}\right|}\right\} = 3$.

\begin{figure}
\centering
\centerline{\epsfig{figure=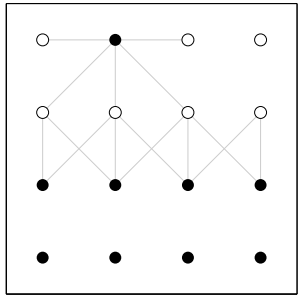,scale=.75}}
\caption{An application of the grid inequality with $m=4$, $\mathcal{A}$ is the set of black dots, $\mathcal{B}$ is the set of white dots, $|\mathcal{A}|=9$, $|\mathcal{B}|=7$, $\partial_{\mathcal{A},\mathcal{B}}=15$}
\label{fig:isoperimetry}
\end{figure}

\section{Lower Bound on Cut Set Capacity}
\label{relatingtoadhoc}

\subsection{Main Results}
In this section, we argue that the $\epsilon$-niceness property of the induced random geometric graph of the network, together with the grid inequality stated in section~\ref{gridinequalities} yield a lower bound on the cut-set capacity $C(\mathcal{V}_\mathcal{S})$ of broadcast wireless networks. The following theorem identifies this lower bound as a function of the number of nodes $n$, transmission radius $r_n$ and erasure probabilities.

\begin{theorem}
\label{thm:capcut}
Consider the setting of a wireless ad hoc erasure relay network, with $n$ nodes $\mathcal{V}_n$ independently and uniformly distributed on the unit square $[0,1]^2$. Let $\epsilon \in (0,\frac{1}{5})$ and let $\left\{r_n\right\}_{n=1}^{\infty}$ be a sequence of positive radii such that $\lim_{n\rightarrow\infty}r_n=0$ and $r_n \geq \frac{6}{\epsilon}\sqrt{\frac{\log{n}}{n}}$ for $n$ large enough. Let $\mathcal{G}_n = \mathcal{G}(\mathcal{V}_n,r_n)$ be the corresponding sequence of random geometric graphs. Let $\mathcal{V}_\mathcal{S}$ be an arbitrary $\mathcal{S}$-cut such that $|\mathcal{V}_\mathcal{S}|=\alpha n$ for some $\alpha > 0$. Then $\mathcal{G}_n$ is almost surely $\epsilon$-nice and if $\gamma_{uv} = \gamma~\forall~(u,v)\in\mathcal{E}$ is the erasure probability across all connected links, then, for sufficiently large $n$, the capacity cut set bound $C(\mathcal{V}_\mathcal{S})$ is lower bounded by the following expression:
\begin{equation}
\label{lowerbound}
C(\mathcal{V}_\mathcal{S}) \geq \frac{3}{2} \min\left\{\sqrt{\alpha-2\epsilon},\sqrt{1-\alpha-2\epsilon}\right\}
\frac{1-2\epsilon}{\sqrt{1+\epsilon}}nr_n \left( 1 - \gamma^{(1-2\epsilon)\frac{1}{4}n r_n^2}\right)
\end{equation}
\end{theorem}

\begin{proof}
The proof proceeds with a similar spirit to the proof of Lemma 5.2 in~\cite{Dia01}. For a full proof, see the Appendix.
\end{proof}

So far, we have shown that if we have a cut between  a constant fraction $\alpha_1 n$ of the nodes on one side and another fraction $\alpha_2 n$ of the nodes on the other, then due to the network topology described in Section~\ref{assumptions} the capacity cut-set evaluated at that particular cut cannot be asymptomatically smaller than a multiple of $n r_n$, as indicated in~(\ref{lowerbound}). To characterize the capacity of the entire network however, all cuts separating the source(s) and destination(s) have to be considered. Due to the special structure of random geometric graphs and Bernoulli point processes and their resemblance to grid, intuitively, for the single source, single destination case, the minimizing cut will almost surely put only a constant number of nodes (the source node and possibly a constant number of nodes it is connected to) on one end and all remaining nodes on the other end, assuming equal erasure probabilities across all links. Since such cuts are not balanced, that is, do not separate a constant fraction of the nodes from the others,~(\ref{lowerbound}) is not immediately applicable. However, it is not intuitive what the minimizing cut is when the network has multiple sources and multiple destinations spread arbitrarily in the unit square. The following theorem generalizes the bound we just derived in the multiple sources, multiple destinations case. Intuitively, when each of the number of source and destination nodes is a constant fraction of all the nodes, the minimizing cut-set will also be balanced and hence the bound derived in Theorem~\ref{thm:capcut} applies.

\begin{theorem}
\label{thm:mainthm}
In the network setup described above, assume there are $\alpha_1 n$ source nodes and $\alpha_2 n$ destination nodes. Assume that each of the source nodes and destination nodes can communicate among each other via incapacitated, error free links. Then the total end-to-end throughput $C_{broadcast}$ is lower bounded by
$$C_{broadcast} \geq \frac{3}{2} \sqrt{\alpha-2\epsilon}
\frac{1-2\epsilon}{\sqrt{1+\epsilon}}nr_n \left( 1 - \gamma^{(1-2\epsilon)\frac{1}{4}n r_n^2}\right)$$
where $\alpha = \min\{\alpha_1, \alpha_2, 1-\alpha_1, 1-\alpha_2\}$.
\end{theorem}

\begin{figure}
\centering
\centerline{\epsfig{figure=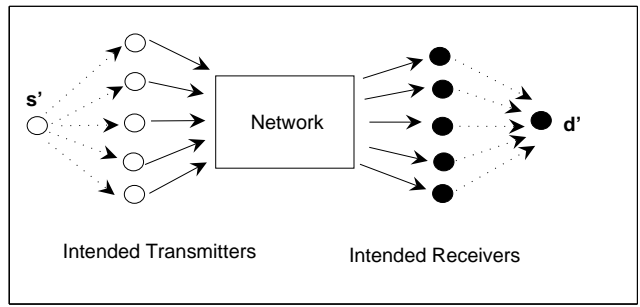,scale=.75}}
\caption{Illustrating the addition of a theoretical source and destination nodes}% that can communicate with the set of sources and destinations via incapacitated error free links (dashed).}
\label{fig:idealntwk}
\end{figure}

\begin{proof}
 The capability of the source nodes to communicate among each others via incapacitated, error free links can be modeled by adding a theoretical source node $s'$ which connects to all $\alpha_1 n$ source nodes via incapacitated error free links and a theoretical destination node $d'$ to which all $\alpha_2 n$ destination nodes connect via incapacitated error free links, as shown in Fig.~\ref{fig:idealntwk}. Then apply the single source, single destination capacity cut set bound of section~\ref{capCutSet} and consider the cut $\mathcal{V}_\mathcal{S}$ that minimizes the capacity expression in equation~\ref{eqn:capacity}. Any constraints from cuts that cut through the incapacitated links cannot be tight as the capacity of an incapacitated link can be made arbitrarily large. Therefore, the cut defining the capacity $\mathcal{V}_\mathcal{S}$ must be such that all original $\alpha_1 n$ source nodes are on one side and all original $\alpha_2 n$ destination nodes are on the other side. If $|\mathcal{V}_\mathcal{S}|=\alpha' n$, then $\min\{\alpha_1,1-\alpha_1,\alpha_2,1-\alpha_2\} < \min\{\alpha',1-\alpha'\}$. The result then follows.
\end{proof}
%\vspace{1em}
\subsection{Scaling Laws}
We have thus established a scaling law of $\Theta(n r_n)$ in the case of wireless broadcast erasure networks when the sources and destinations each form a constant fraction of the nodes. The effect of the erasure probabilities is not very significant in the lower bound we derived, at least when it is constant (and not a function of $n$) and when $n$ is large. In that case, for any nontrivial erasure probability $\gamma$ ($\gamma<1$) , its effect to the lower bound established in theorem~\ref{thm:mainthm} can be made arbitrarily small for sufficiently large $n$, since $\lim_{n\rightarrow\infty} \gamma^{(1-2\epsilon)\frac{1}{4}nr_n^2} = 0$ for $r_n \geq \frac{6}{\epsilon} \sqrt{\frac{\log{n}}{n}}$. For the critical value of $r_n$ which scales as $\sqrt{\frac{\log{n}}{n}}$, the proven lower bound scales as $\sqrt{n\log{n}}$, which agrees up to a $\sqrt{\log{n}}$ factor with the $\sqrt{n}$ scaling law shown in~\cite{GupKum00} although the models are different.\\

When the transmission radius $r_n$ scales as $\sqrt{\frac{\log{n}}{n}}$, we observe a diminishing law of returns in throughput as more nodes are added to the network. This entails a significant increase in end-to-end throughput upon adding a new node only when the number of nodes $n$ is relatively small. In fact, if the number of nodes is sufficiently small, the induced random geometric graph might even be disconnected. We identify regimes with a relatively small number of nodes as being \emph{power limited}. In such networks, increasing the transmission power of each node increases throughput significantly. As more and more nodes are added to the network, we expect that the network becomes \emph{interference limited} whereby increasing the power of all the nodes no longer yields significant gains. Fig.~\ref{fig:regimes} illustrates these two regimes. Although the issue of interference is not directly addressed in this paper, one can model that indirectly by a suitable choice of the transmission radius $r_n$ proportional to $n$, and erasure probability $\gamma_n$ that increases with $n$. As the number of nodes increase, more opportunities are provided to transmit the messages from source nodes to destination nodes, hence increasing the end-to-end throughput, but each node's transmission radius decreases and the links' erasure probabilities increases, thus limiting the net throughput gain.

\begin{figure}
\centering
\centerline{\epsfig{figure=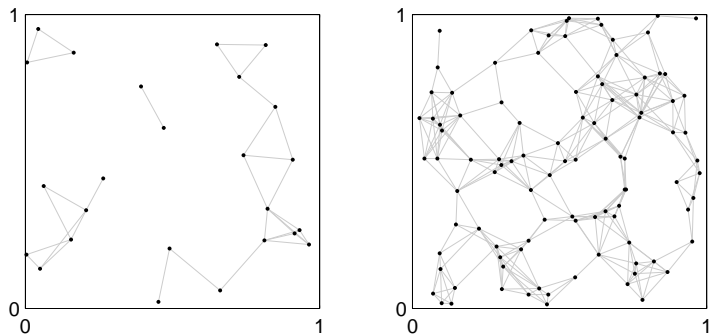,scale=.75}}
\caption{A power limited network (to the left) and an interference limited network (to the right).}
\label{fig:regimes}
\end{figure}

\subsection{Tightness of Lower Bound}
\label{sec:tight}
We now argue that the lower bound presented above is tight, in the sense that for every $n$, there exists a choice (actually many choices) of source and destination nodes that would yield a network capacity of at most $\Theta(nr_n)$. One simple example is to assign all nodes in the left rectangle, i.e. with abscissa smaller than or equal to 0.5, as source nodes and all nodes in the right rectangle, i.e. with abscissa larger than 0.5, as destination nodes as illustrated in Fig.~\ref{fig:tight}. To see this, we dissect the unit square into square cells each of side length $r_n$. Consider the cut $\mathcal{V}_\mathcal{S} = \left\{x \textrm{ s.t. } x_1 \leq 0.5 \right\}$ where $x_1$ denotes the abscissa of $x$. By $\epsilon$-niceness, each cell contains at most $(1+\epsilon)n r_n^2$ nodes for any $\epsilon\in(0,1)$. Since the side length of the cell is $r_n$, if two nodes are connected, then they must be neighbors. Then, the following cut-set bound follows:
\begin{equation}
C(\mathcal{V}_\mathcal{S}) \leq \frac{1}{r_n} nr_n^2\left(1+\epsilon\right)\left(1-\gamma^{3(1+\epsilon)nr_n^2}\right) = \Theta(nr_n)
\end{equation}
We have thus shown that the lower bound is tight. This example illustrates a simple design principle: if the node locations in a network are i.i.d. distributed uniformly in the unit square but the network designer could chose which nodes are to be source nodes and which are to be destination nodes, then it is best to ``scatter'' the sources and destinations rather than ``clutter'' all source nodes together and all destination nodes together as shown in Fig.~\ref{fig:tight}. If source nodes and destination are paired up, as suggested in Fig.~\ref{fig:loose} for example, then a linear capacity scaling would be possible. This example is illustrative but not very practical because in this scenario information between sources and destination is confined within small ranges of order $\Theta(r_n)$ and is not really ``transported''.

\begin{figure}
\centering
\centerline{\epsfig{figure=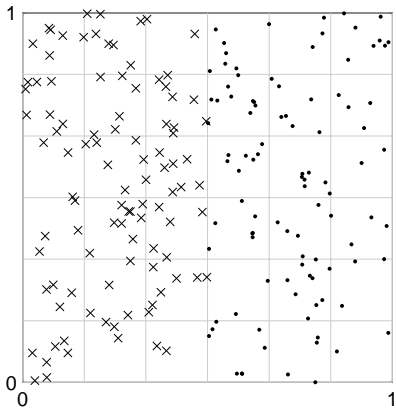,scale=.75}}
\caption{Illustrating a choice of source and destination nodes such that network capacity = $\Theta(nr_n)$.}
\label{fig:tight}
\end{figure}

\begin{figure}
\centering
\centerline{\epsfig{figure=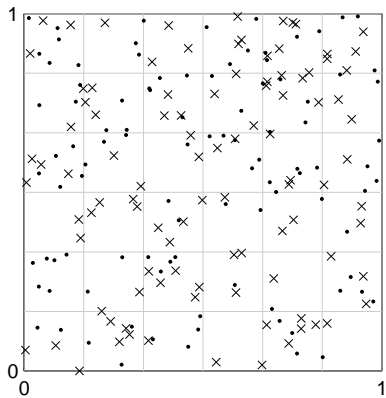,scale=.75}}
\caption{Illustrating a choice of source and destination nodes yielding capacity of order $\Theta(n)$.}
\label{fig:loose}
\end{figure}

\section{Multicast}
\label{sec:multicast}
We now investigate relaxing the broadcast constraint on the nodes, i.e. the constraint that each node (including relay nodes) must transmit the same signal on all its outgoing edges. This will yield a gain in the end-to-end capacity. However, this gain will greatly depend on $r_n$ and the erasure probabilities. We expect the gain of multicast to be more pronounced for large transmission radii since then the node will be able to communicate with more nodes and transmit more distinct messages across the network. Also, this gain will be more apparent for small erasure probabilities. Intuitively, if transmissions are very unlikely to be successful and the expected number of successful transmissions for each node is only a constant, say one, then we do not expect significant gains from relaxing the broadcast constraint. The following theorems formalizes this gain.

\begin{theorem}
\label{thm:relaxingbroadcast}
If the broadcast constraint is relaxed and each node is allowed to send possibly distinct messages across its outgoing links, then the lower bound of the capacity $C_{multicast}$ of the network is:
$$C_{multicast} \geq \frac{3}{8} \sqrt{\alpha-2\epsilon}
\frac{(1-2\epsilon)^2}{\sqrt{1+\epsilon}}n^2 r_n^3 (1-\gamma)$$
where $\alpha = \min\{\alpha_1, \alpha_2, 1-\alpha_1, 1-\alpha_2\}$.
\end{theorem}
\begin{proof}
~\cite{Ahl00,Li03,KoeMed03} prove that in such \emph{wireline} networks, the corresponding capacity cut set bound is also tight:  $\sum_{u\in \mathcal{V}_\mathcal{S}}\sum_{v: (u,v)\in \mathcal{[\mathcal{V}_\mathcal{S}:\mathcal{V}_\mathcal{D}]}} (1-\gamma_{uv})$. Applying a similar analysis on this cut set bound as above yields the result.
\end{proof}

If we denote by $C_{BC}$ the lower bound proved above for the broadcast network and by $C_{MC}$ that proved for the wireline network, then the apparent gain is given by:
$$G_{MC}\triangleq\frac{C_{MC}}{C_{BC}} = \frac{1}{4}(1-2\epsilon)nr_n^2 \frac{1-\gamma}{1 - \gamma^{(1-2\epsilon)\frac{1}{4}n r_n^2}}$$

Evaluating this gain at the critical value of transmission radius $r_n = c \sqrt{\frac{\log{n}}{n}}$ for sufficiently large $c$, we get:
\begin{equation}
\label{eqn:Gmc}
G_{MC} = \frac{1}{4}(1-2\epsilon)c^2 \log{n} \frac{1-\gamma}{1 - \gamma^{(1-2\epsilon)\frac{1}{4} c^2\log{n}}}
\end{equation}

If the erasure probability $\gamma$ is constant and not a function of $n$, the asymptotic gain is a $\log{n}$ factor. This is not surprising because each node is connected to at most a constant factor of $\frac{1}{4}nr_n^2 = \frac{1}{4}c^2 \log{n}$ and hence the throughput cannot increase beyond that factor due to the ability of multicast. However, if the erasure probability $\gamma$ increases with $n$ as $\gamma_n = 1-\frac{1}{g(n)\log{n}}$, where $g(n)$ is some positive real valued function, then the gain $G_{MC}$ due to multicast scales depending on the asymptotic behavior of $g(n)$ as formalized in lemma below:

\begin{lemma}
If the erasure probability $\gamma_n$ scales as $\gamma_n = 1-\frac{1}{g(n)\log{n}}$ where $g(n)$ is a positive real valued function and $G_{MC}$ scales with $\gamma_n$ as in equation~(\ref{eqn:Gmc}), then:
\begin{enumerate}
\item
If $\lim_{n\rightarrow\infty}g(n)=0$, then we observe a significant gain $G_{MC}$: $\lim_{n\rightarrow\infty}G_{MC} = \infty$.
\item
If $\lim_{n\rightarrow\infty}g(n)$ exists and is non-zero, or is infinite, then the gain $G_{MC}$ is at most a constant.
\end{enumerate}
\end{lemma}

\begin{proof}
Letting $c_1 = \frac{1}{4}(1-2\epsilon)c^2$:
\begin{eqnarray}
\lim_{n\rightarrow\infty}G_{MC} &=& \lim_{n\rightarrow\infty}  c_1 \log{n} \frac{1-\gamma_n}{1 - \gamma^{c_1\log{n}}_n} \nonumber\\
    &=& \lim_{n\rightarrow\infty} \frac{\frac{c_1 }{g(n)}}{1 - \left(1-\frac{1}{g(n)\log(n)}\right)^{c_1\log{n}}}\nonumber\\
    \label{explimit}
    &=& \lim_{n\rightarrow\infty} \frac{\frac{c_1}{g(n)}}{1 - \exp\left(-\frac{c_1}{g(n)}\right)}
\end{eqnarray}
Equation~(\ref{explimit}) uses the identity $\lim_{n\rightarrow\infty}\left(1+\frac{z}{n}\right)^n = e^{z}$.\\
Hence, if $\lim_{n\rightarrow\infty}g(n) = 0$, then $\lim_{n\rightarrow\infty}G_{MC} = \infty$.\\
Otherwise if $g(n)$ satisfies condition 2) above, then using the second order Taylor series approximation $e^{-x} \leq 1 - x + \frac{1}{2!}x^2$ for $x\geq0$, we obtain:
\begin{eqnarray}
\lim_{n\rightarrow\infty}G_{MC}
    &\leq& \lim_{n\rightarrow\infty} \frac{\frac{c_1}{g(n)}}{\frac{c_1}{g(n)}-\frac{1}{2!}\frac{c_1^2}{g^2(n)}}\\
    &=& \lim_{n\rightarrow\infty} \frac{1}{1-\frac{1}{2!}\frac{c_1}{g(n)}}\\
    &=& \textrm{constant}
\end{eqnarray}
\end{proof}

It might seem surprising at first that relaxing the broadcast constraint did not enhance the throughput (lower bound) by more than a constant when $\gamma_n = 1-\frac{1}{g(n)\log{n}}$ and $g(n)$ satisfies condition 2) above, but this result illustrates the robustness factor of the broadcast network when the erasure probability is very high. This aligns with our intuition. Since each node is connected to at most a constant fraction of $\log{n}$ and the success probability $1-\gamma_n$ is equal to $\frac{1}{g(n)\log{n}}$, the expected number of successful transmissions for a node is proportional to $\log{n}\times \frac{1}{g(n)\log{n}}$, which is at most a constant. Hence, as argued earlier, at most a constant improvement gain can be achieved by multi-cast in this case.

\section{Random Erasure Probabilities}
So far, all erasure probabilities have been assumed fixed and equal across all links in the network. This is unlikely to be true in a real wireless network due to fading, interference and node geometry. Even if the erasure probabilities were close to being fixed, it might be unpractical to characterize each link separately, especially in large networks, because there are many links. We suggest modeling erasure probabilities themselves as random variables and explore the impact of this additional uncertainty on the performance of the network. Assigning random erasure probabilities that tend to increase as the number of nodes increases, can partially account for fading and interference. A similar approach was adopted in~\cite{Smi07},~\cite{Smi06} whereby the non-erasure probabilities, i.e. success probabilities decay polynomially with distance.\\

Intuitively, since each node is connected to a multiple of $nr_n^2$ other nodes, which is of the order of at least $\log{n}$ for $r_n \geq c \sqrt{\frac{\log{n}}{n}}$, we expect that by the law of large numbers, a suitable average of the erasure probabilities across those outgoing links is what matters. The following lemma formally characterizes a lower bound that is analogous to the one derived in Theorem~\ref{thm:mainthm} when the erasure probabilities are random.

\begin{lemma}
\label{lemma:randomearsures}
If the erasure probabilities $\gamma_{ij}$ are identically distributed, pairwise independent random variables with the same distribution as $\gamma$, such that $\log{\gamma}$ has finite mean and variance then the following is a lower bound on the broadcast capacity cut with high probability\footnote{That is probability goes to 1 as $n$ goes to $\infty$}:
$$C_{BC,var} \geq \frac{3}{2} \sqrt{\alpha-2\epsilon}
\frac{1-2\epsilon}{\sqrt{1+\epsilon}}nr_n \left( 1 -  \exp\left(m\left(E\log{\gamma}+\epsilon \right)\right)\right)$$
where $\alpha = \min\{\alpha_1, \alpha_2, 1-\alpha_1, 1-\alpha_2\}, m=(1-2\epsilon)\frac{1}{4}n r_n^2$ for any $\epsilon\in(0,\frac{1}{5})$.
\end{lemma}

\begin{proof}
The cut set capacity bound proved in Theorem~\ref{thm:mainthm} can now be thought of as a random variable.
Ordering the $(1-2\epsilon)\frac{1}{4}n r_n^2$ terms of the $\gamma_{uv}$'s appearing in the result of Theorem~\ref{thm:mainthm} arbitrarily as $\gamma^{(1)}, \gamma^{(2)},\ldots \gamma^{(m)}$, where where $m=(1-2\epsilon)\frac{1}{4}n r_n^2\geq (1-2\epsilon)\frac{1}{4}c^2\log{n}$. We can express this random variable as:
$$\frac{3}{2} \min\left\{\sqrt{\alpha-2\epsilon},\sqrt{1-\alpha-2\epsilon}\right\}
\frac{1-2\epsilon}{\sqrt{1+\epsilon}}nr_n \left( 1 - \prod_{i=1}^{m} \gamma^{(i)}\right)$$  or equivalently as:
$$\frac{3}{2} \min\left\{\sqrt{\alpha-2\epsilon},\sqrt{1-\alpha-2\epsilon}\right\}
\frac{1-2\epsilon}{\sqrt{1+\epsilon}}nr_n \left( 1 - \exp\left\{\sum_{i=1}^{m} \log{\gamma^{(i)}}\right\}\right)$$
\begin{eqnarray*}
\Pr\left\{\frac{1}{m}\sum_{i=1}^{m} \log{\gamma^{(i)}} > \left(E\log{\gamma^{(1)}}+\epsilon \right)\right\} &\leq& \Pr\left\{\left|\left(\frac{1}{m}\sum_{i=1}^{m}\log{\gamma^{(i)}}\right)  - E\log{\gamma^{(1)}}\right| > \epsilon  \right\} \\
&\leq& \frac{\textrm{var}\left(\frac{1}{m}\sum_{i=1}^{m}\log{\gamma^{(i)}}\right)}{\epsilon^2}\\
&=& \frac{\textrm{var}\left( \log{\gamma^{(1)}} \right)}{m \epsilon^2}\\
&\rightarrow& 0 \textrm{ as $n\rightarrow\infty$}.
\end{eqnarray*}
Therefore:
\begin{eqnarray*}
\Pr\left\{\exp\left(\sum_{i=1}^{m} \log{\gamma^{(i)}}\right) > \exp\left(m\left(E\log{\gamma^{(1)}}+\epsilon \right)\right)\right\} &=&
\Pr\left\{
\prod_{i=1}^{m} \gamma^{(i)} > \exp\left(m\left(E\log{\gamma^{(1)}}+\epsilon \right)\right)\right\}\\
&\rightarrow& 0 \textrm{ as $n\rightarrow\infty$}
\end{eqnarray*}
\end{proof}
Notice that by concavity of the logarithm function and by Jensen's inequality, $E\log{\gamma^{(1)}} \leq \log{E{\gamma^{(1)}}}$. Hence, by comparing the expressions of Lemma~\ref{lemma:randomearsures} to that of Theorem~\ref{thm:mainthm}, it follows that there is actually a gain in the case where $\gamma^{(1)}$ is a random variable with distribution same as $\gamma$ over that where $\gamma^{(1)}$ is fixed and is equal to the mean $E\gamma$. By comparing the lower bounds, we notice a gain of
\begin{equation}
\label{eqn:gainvar}
G_{var} \triangleq \frac{1-\exp{\left(m(E\log\gamma+\epsilon)\right)}}{1-(E\gamma)^m}
\end{equation}
 due to the variability of the erasure probabilities where $m=(1-2\epsilon)\frac{1}{4}nr_n^2$.
We will demonstrate an example to illustrate this gain due to the variability in erasure probabilities. We will consider two cases. The first case is that of a fixed erasure probability $\gamma_1 = 0.5$. In the second case, $\gamma_2$ is a random variable that is uniform over $[0.25,0.75]$. In the latter case, $$E\log{\gamma_2} = \frac{1}{0.5}\left(\left(0.75\log{0.75} - 0.75\right) - \left(0.25\log{0.25} - 0.25\right)\right) \approx -0.7384$$
Since $\epsilon$ can be made arbitrarily small, say $\epsilon=0.01$. Then, the gain is about $\frac{1-0.4827^m}{1-0.5^m}$, which is greater than one.\\

This result might be surprising because variability and unequal factors usually yield a loss. For example, the capacity of an additive white gaussian channel is proportional to $\log{(1+SNR)}$ where $SNR$ is the signal to noise ratio. For a fixed noise level and a fixed transmit average power, varying the transmit power yields a loss in this case. Alternatively, for a fixed transmit power, varying the noise level yields a loss. It follows by the concavity of the function $f(x)=\log{(1+x)}$ and Jensen's that $$E\left[{\log{(1+SNR)}}\right] \leq \log{(1+E{(SNR)})}$$ and this shows that variability cannot yield any gain in this case. So practically, if a system is operating at a certain power level, then decreasing $SNR$ decreases channel capacity significantly but increasing $SNR$ by the same amount yields a smaller gain. The situation is different in the case of variability of erasures in a broadcast wireless networks since only one successful transmission across the cut is sufficient to ``transport'' the bit from one side of the cut to the other. More successful transmissions do not increase the capacity of that particular cut. We conclude that variability in erasure probabilities provides statistical diversity that improves the chances of at least one successful transmission.

\section{Conclusion}
\label{sec:conclusion}
We analyzed the cut-set capacities in the canonical framework of $n$ nodes uniformly
and independently distributed in the unit square whereby nodes are connected to nearby nodes that lie
within their transmission radius. The core of the analysis is based on random geometric graph theory and
its analogy with that of the deterministic grid. A lower bound on the end-to-end throughput in the case of
arbitrary multiple sources and multiple destinations was presented and a scaling law of $n r_n$ was observed in the case where the broadcast constraint is enforced.
This lower bound agrees with the $\sqrt{n}$ scaling law presented in~\cite{GupKum00}
when the transmission radius $r_n$ scales as $\sqrt{\frac{\log{n}}{n}}$ although the models are different.
The lower bound derived in this paper reflects the effect of physical layer parameters, such as erasure probabilities
and transmission radii of the nodes.
We investigated relaxing the broadcast constraint and proved a lower bound that scales as $n^2 r_n^3$ when nodes are allowed to send distinct messages across their outgoing links, assuming constant erasure probabilities\footnote{That is, not a function of $n$}. We also concluded that multicast allows a significant gain in capacity only when the expected number of successful transmissions is large. Although we did not explicitly deal with interference and fading, we allowed the erasure probability $\gamma$ to be a function of the number of nodes $n$. Hence, interference can partially be accounted for by modeling the erasure probability as an increasing function with $n$, say as $1-\frac{1}{\log{n}}$. Similarly, fading can be accounted for by assuming that the erasure probabilities are random variables. We finally showed somewhat surprisingly that this variability can actually boost the end-to-end throughput for large networks.

\appendix
{\bf Proof of Theorem~\ref{thm:capcut}}
\begin{proof}
Note that given $\gamma_{uv} = \gamma$ for all $(u,v)\in\mathcal{E}$, the capacity expression simplifies to:
$$C(\mathcal{V}_\mathcal{S}) = \sum_{u\in\mathcal{V}_\mathcal{S}^*}\left(1-\gamma^{|N_O(u)|}\right)$$
Also, given the assumptions on $\mathcal{V}_n$ and $r_n$, it follows that $\mathcal{G}_n$ is $\epsilon$-nice almost surely by theorem~\ref{thm:epsilonnice}.\\
We will neglect ceiling and floors for simplicity. Dissect the unit area square into $4/r_n^2$ smaller square cells. We color the nodes in $\mathcal{V}_\mathcal{S}$ white and those in $\mathcal{V}_\mathcal{D}$ black. We also color the cells as follows. We color each cell black if it contains at most $\frac{1}{5}\epsilon n r_n^2$ white points, white if it contains at most $\frac{1}{5}\epsilon n r_n^2$ black points and grey otherwise. This can be thought of as clumping all the nodes in each cell to one super-node, having the same color as that assigned to the cell. Less formally, we color the super node the color of the majority of the points it represents. Hence black and white cells denote cells with mostly black and white points respectively. Grey cells are ``mixed'' and have many points from both colors.
We consider two cases, depending on the number of grey cells, $G_n$.
The following two lemmas formally prove that in both cases, the lower bound stated in the theorem is valid.
Intuitively, if there are many grey cells, then there must be many edges in the cut-set due to the edges between white and black points in each grey cell and these edges will be enough to establish the lower bound. On the other hand, if there were very few grey cells, we will recolor the grey cells pessimistically, and apply the grid inequality established in section~\ref{gridinequalities} on the super nodes.
\end{proof}

\begin{lemma}
If $G_n \geq \frac{25}{\epsilon r_n}$, then $C(\mathcal{V}_\mathcal{S})\geq 5 n r_n \left( 1 - \gamma^{\frac{1}{5}\epsilon n r_n^2} \right)$
\end{lemma}
\begin{proof}
By the definition of a grey cell, each grey cell contains at least $\frac{1}{5}\epsilon n r_n^2$ points in $\mathcal{V}_\mathcal{S}$ and $\frac{1}{5}\epsilon n r_n^2$ points in $\mathcal{V}_\mathcal{D}$. Nodes in the same cell are certainly connected because the side length of the cell is $r_n/2$. Therefore, each grey cell contains at least $\frac{1}{5}\epsilon n r_n^2$ points in $\mathcal{V}_\mathcal{S}^*$, each of which has an out-degree of at least $\frac{1}{5}\epsilon n r_n^2$. Considering only these edges within grey cells, each grey cell contributes at least $\frac{1}{5}\epsilon n r_n^2 \left( 1 - \gamma^{\frac{1}{5}\epsilon n r_n^2} \right)$ to $C(\mathcal{V}_\mathcal{S})$. Hence, $C(\mathcal{V}_\mathcal{S}) \geq  G_n \frac{1}{5}\epsilon n r_n^2 \left( 1 - \gamma^{\frac{1}{5}\epsilon n r_n^2} \right)$ yielding the lemma.
\end{proof}

\begin{lemma}
\label{lemma:app}
If $G_n < \frac{25}{\epsilon r_n}$, then $$C(\mathcal{V}_\mathcal{S}) \geq \frac{3}{2} \min\left\{\sqrt{\alpha-2\epsilon},\sqrt{1-\alpha-2\epsilon}\right\}
\frac{1-2\epsilon}{\sqrt{1+\epsilon}}nr_n \left( 1 - \gamma^{(1-2\epsilon)\frac{1}{4}n r_n^2}\right)$$
\end{lemma}

\begin{proof}

For this case, we will only consider the edges $(u,v)$ contributing to the capacity cut such that $u$ and $v$ belong to distinct but neighboring cells, including diagonals. We will first show that we can recolor all nodes in a cell to all black or all white without increasing the value of the capacity cut when restricted to edges between neighboring cells. To see that, consider a cell with $t$ nodes, $w$ of which are white and $(t-w)$ are black. Assume its neighboring cells has $w'$ white nodes and $b'$ black nodes. Then the capacity cut, restricted to the edges described above, has the form $$w(1-\gamma^{b'})+w'(1-\gamma^{t-w+c_1}) + c_2$$ where $c_1$ and $c_2$ are constants independent of $w$ and $t$. With $w', b'$ and $t$ fixed, the above expression is a concave function in $w$. Indeed, the second derivative of this continuous function with respect to $w$ is $-w'\gamma^{t-w+c_1}$ which is non-positive. So its minimum over the compact set $w\in[0,t]$ is achieved at an extreme value. This proves that we can color all nodes of each square as all white or all black without increasing the capacity cut when restricted to edges crossing neighboring cells. We recolor nodes in each grey cell to all white or all black whichever does not increase the value of the capacity cut (restricted to edges crossing neighboring cells). We would have hence eliminated all grey cells. Since we recolored the nodes pessimistically, the lower bound that we establish now will still hold for the original case.\\

Let $W_n$ and $B_n$ denote the number of white cells and black cells respectively after recoloring. The number of points whose color has been changed is at most all the points in all grey cells, which can be bounded by $\epsilon n$ for sufficiently large $n$ as prove below:
 \begin{eqnarray}
 \textrm{All the points in all grey cells} &\leq& G_n (1+\epsilon)\frac{1}{4}nr_n^2 \label{allrecolor1}\\
            &\leq& \frac{25}{\epsilon r_n} (1+\epsilon)\frac{1}{4}nr_n^2 \label{allrecolor2}\\
            &=& \left(\frac{25}{4} \frac{1+\epsilon}{\epsilon}r_n \right)n \label{allrecolor3}\\
            &\leq& \epsilon n \textrm{ for $n$ large enough} \label{allrecolor4}
 \end{eqnarray}
 Equation (\ref{allrecolor1}) follows from $\epsilon$-niceness, (\ref{allrecolor2}) follows from the bound on $G_n$, and the last inequality (\ref{allrecolor4}) follows from the assumption that $\lim_{n\rightarrow \infty} r_n = 0$.\\

Before recoloring, the number of white points was exactly $\alpha n$ and at most $\epsilon n$ of them were recolored. Therefore, the number of white points after recoloring is at least $(\alpha - \epsilon)n$ for large $n$. By construction of black cells, the number of white points in black cells is at most $\frac{1}{5}\epsilon n r_n^2 4 \lceil 1/r_n\rceil^2 \leq \epsilon n$. Therefore, the number of white points in white cells after recoloring is at least $(\alpha-2\epsilon)n$. It follows by $\epsilon$-niceness then, that $W_n$ is at least $\frac{(\alpha-2\epsilon)n}{\frac{1}{4}(1+\epsilon)n r_n^2} = \frac{4(\alpha-2\epsilon)}{(1+\epsilon)r_n^2}$. Similarly, $B_n \geq \frac{4(1-\alpha-2\epsilon)}{(1+\epsilon)r_n^2}$.\\

Denote by $\delta \mathcal{G}$ the length of the boundary between white cells and black cells, that is, the number of distinct pairs of neighboring cells of opposite colors after recoloring. An application to the isoperimetric inequality stated in theorem~\ref{eqn:isoperimetry} yields:
\begin{equation}
\delta \mathcal{G} \geq 3 \min \left\{ \sqrt{W_n}, \sqrt{B_n} \right\} \geq \frac{6}{r_n} \min\left\{\sqrt{\frac{\alpha-2\epsilon}{1+\epsilon}},\sqrt{\frac{1-\alpha-2\epsilon}{1+\epsilon}}\right\}
\end{equation}

But each white cell contains at least $(1-2\epsilon)\frac{1}{4}nr_n^2$ white points and each black cell contains at least $(1-2\epsilon)\frac{1}{4}nr_n^2$ black points. Thus
\begin{equation}
C(\mathcal{V}_\mathcal{S}) \geq \delta \mathcal{G} (1-2\epsilon)\frac{1}{4}nr_n^2 \left( 1 - \gamma^{(1-2\epsilon)\frac{1}{4}n r_n^2}\right)
\end{equation}
Therefore,
\begin{equation}
C(\mathcal{V}_\mathcal{S}) \geq \frac{3}{2} \min\left\{\sqrt{\alpha-2\epsilon},\sqrt{1-\alpha-2\epsilon}\right\}
\frac{1-2\epsilon}{\sqrt{1+\epsilon}}nr_n \left( 1 - \gamma^{(1-2\epsilon)\frac{1}{4}n r_n^2}\right)
\end{equation}

\end{proof}

%\bibitem{mobility}
%M. Grossglauser and D. N. C. Tse, ``Mobility Increases the Capacity of
%Ad Hoc Wireless Networks,'' \emph{IEEE/ACM Trans. Netw.}, vol. 10, no. 4, pp.
%477--486, Aug. 2002.

\bibliography{Andrews}
\bibliographystyle{IEEEtran}

\end{document}